\documentclass{llncs}

\usepackage{hyperref} 
\usepackage{color}
\usepackage{amssymb}

\renewcommand{\tt}{\ttfamily}
\newcommand{\codefont}{\tt\small}
\newcommand{\code}[1]{\mbox{\codefont{#1}}}
\newcommand{\ccode}[1]{``\code{#1}''}
\newcommand{\us}{\char95} 

\usepackage{listings}
\lstnewenvironment{curry}[1][]
  {\lstset{literate={->}{{$\rightarrow{}\!\!\!$}}3,#1}}{}

\newcommand{\Ac}{{\cal{A}}} 
\newcommand{\Cc}{{\cal{C}}} 
\newcommand{\Fc}{{\cal{F}}} 

\newcommand{\dom}[1]{\ensuremath{\mathit{dom}(#1)}}
\newcommand{\var}[1]{\ensuremath{\mathit{Var}(#1)}}
\newcommand{\io}{\ensuremath{\mathit{io}}}

\renewcommand{\emptyset}{\varnothing}
\renewcommand{\epsilon}{\varepsilon}
\newcommand{\iomapsto}{\hookrightarrow} 
\newcommand{\IO}{\mathit{IO}}
\newcommand{\CT}{\mathit{CT}}
\newcommand{\iosubst}[3]{(#1,#2,#3)} 
\newcommand{\rulename}[1]{\mbox{\sf #1}}
\newcommand{\vrulename}[1]{\mbox{\sf #1$_{\mathit{nf}}$}}
\newcommand{\seq}[1]{\overline{#1}} 

\begin{document}
\pagestyle{plain}

\title{Inferring Non-Failure Conditions\\ for Declarative Programs}

\author{
Michael Hanus
}
\institute{
Institut f\"ur Informatik, Kiel University, Germany\\
\email{mh@informatik.uni-kiel.de}
}

\maketitle

\begin{abstract}
Unintended failures during a computation are painful but
frequent during software development.
Failures due to external reasons (e.g., missing files, no permissions)
can be caught by exception handlers.
Programming failures, such as calling a partially defined
operation with unintended arguments, are often not caught
due to the assumption that the software is correct.
This paper presents an approach to verify such assumptions.
For this purpose, non-failure conditions for operations are
inferred and then checked in all uses of partially defined operations.
In the positive case, the absence of such failures is ensured.
In the negative case, the programmer could adapt the program
to handle possibly failing situations and check the program again.
Our method is fully automatic and can be applied to larger
declarative programs.
The results of an implementation for functional logic Curry programs
are presented.
\end{abstract}

\section{Introduction}
\label{sec:intro}

The occurrence of failures during a program execution is
painful but still frequent when developing software systems.
The main reasons for such failures are
\begin{itemize}
\item
external, i.e., outside the control of the program,
like missing files or access rights, unexpected formats of external data,
etc.
\item
internal, i.e., programming errors
like calling a partially defined operation with unintended arguments.
\end{itemize}
External failures can be caught by exception handlers
to avoid a crash of the entire software system.
Internal failures are often not caught since they should not occur
in a correct software system.
In practice, however, they occur during software development
and even in deployed systems which results in expensive debugging tasks.
For instance, a typical internal failure in imperative programs
is dereferencing a pointer variable whose current value
is the null pointer (due to this often occurring failure,
Tony Hoare called the introduction of null pointers his
``billion dollar mistake''\footnote{%
\url{http://qconlondon.com/london-2009/speaker/Tony+Hoare}}).

Although null pointer failures cannot occur in declarative programs,
such programs might contain other typical programming errors,
like failures due to incomplete pattern matching.
For instance, consider the following operations (shown in Haskell syntax)
which compute the first element and the tail of a list:
\begin{curry}
head :: [a] -> a               tail :: [a] -> [a]
head (x:xs) = x                 tail (x:xs) = xs
\end{curry}
In a correct program, it must be ensured that \code{head} and \code{tail}
are not evaluated on empty lists.
If we are not sure about the data provided at run time,
we can check the arguments of partial operations before the application.
For instance, the following code snippet
defines an operation to read a command together with some
arguments from standard input (the operation \code{words}
breaks a string into a list of words separated
by white spaces) and calls an operation \code{processCmd}
with the input data:\label{sec:readCmd}
\begin{curry}
readCmd = do putStr "Input a command:"
             s <- getLine
             let ws = words s
             case null ws of True  -> readCmd
                             False -> processCmd$\;$(head ws)$\;$(tail ws)
\end{curry}
By using the predicate \code{null} to check the emptiness of a list,
it is ensured that \code{head} and \code{tail} are not applied to an empty list
in the \code{False} branch of the case expression.

In this paper we present a fully automatic tool which can verify
the non-failure of this program.
Our technique is based on analyzing the types of arguments
and results of operations in order to ensure that partially defined
operations are called with arguments of appropriate types.
The principle idea to use type information for this purpose is not new.
For instance, with \emph{dependent types}, as in
Agda \cite{Norell08}, Coq \cite{BertotCasteran04}, or Idris \cite{Brady13},
or \emph{refinement types}, as in
LiquidHaskell \cite{VazouSeidelJhala14,VazouEtAl14},
one can express restrictions on arguments of operations.
Since one has to prove that these restrictions hold during the
construction of programs, the development of such programs becomes
harder \cite{Stump16}.
Another alternative, proposed in \cite{Hanus18PPDP},
is to annotate operations with \emph{non-fail conditions}
and verify that these conditions hold at each call site
by an external tool, e.g., an SMT solver \cite{deMouraBjorner08}.
In this way, the verification is fully automatic but
requires user-defined annotations and, in some cases,
also the verification of post-conditions or contracts to state
properties about result values of operations \cite{Hanus20FI}.

The main idea of this work is to \emph{infer}
non-fail conditions of operations.
Since the inference of precise conditions is undecidable in general,
we approximate them by \emph{abstract types},
e.g., finite representations of sets of values.
In particular, our contributions are as follows.
\begin{enumerate}
\item 
We define a \emph{call type} for each operation.
If the actual arguments belong to the call type,
the operation is reducible with some rule.
\item
For each operation, we define \emph{in/out types}
to approximate its input/output behavior.
\item
For each call to an operation $g$ occurring in a rule defining $f$,
we check, by considering the call structure and in/out types,
whether the call type of $g$ is satisfied.
If this is not the case, the call type of $f$ is refined
and we repeat the checks with the refined call type.
\end{enumerate}
At the end of this process, each operation has some correct call type
which ensures that it does not fail on arguments belonging to its call type.
Note that the call type might be empty on always failing operations.
To avoid empty call types, one can modify the program code
so that a different branch is taken in case of a failure.

In order to make our approach accessible to various declarative
languages, we formulate and implement it in the
declarative multi-paradigm language Curry \cite{Hanus16Curry}.
Since Curry extends Haskell by logic programming features
and there are also methods to transform logic programs into
Curry programs \cite{Hanus22TPLP}, our approach can also be applied
to purely functional or logic programs.
A consequence of using Curry is the fact that programs might compute with
failures, e.g., it is not an immediate programming error to apply
\code{head} and \code{tail} to possibly empty lists.
However, subcomputations involving such possibly failing calls
must be encapsulated so that it can be checked whether such a computation
has no result (this corresponds to exception handling in
deterministic languages).
If this is done, one can ensure that the overall computation
does not fail even in the presence of encapsulated logic
(non-deterministic) subcomputations.

The paper is structured as follows.
After sketching the basics of Curry in the next section,
we introduce call types and their abstraction in Sect.~\ref{sec:atypes}.
Section~\ref{sec:inouttypes} defines in/out types and
methods to approximate them.
The main section~\ref{sec:inference} presents our method
to infer and check call types for all operations in a program.
We evaluate our approach in Sect.~\ref{sec:evaluation}
before we conclude with a discussion of related work.
The correctness results and their proofs
are contained in the appendix.

\section{Functional Logic Programming and Curry}
\label{sec:flp}

As mentioned above, we develop and implement our method in Curry
so that it is also available for purely functional or logic programs.
Curry \cite{Hanus16Curry} amalgamates features from functional programming
(demand-driven evaluation, strong typing, higher-order functions)
and logic programming
(computing with partial information, unification, constraints),
see \cite{AntoyHanus10CACM,Hanus13} for surveys.
The syntax of Curry is close to Haskell\footnote{%
Variables and function names usually
start with lowercase letters and the names of type and data constructors
start with an uppercase letter. The application of $f$
to $e$ is denoted by juxtaposition (``$f~e$'').}
\cite{PeytonJones03Haskell}.
In addition to Haskell, Curry applies rules
with overlapping left-hand sides in a (don't know) non-deterministic manner
(where Haskell always selects the first matching rule)
and allows \emph{free} (\emph{logic}) \emph{variables} in conditions
and right-hand sides of defining rules.
The operational semantics is based on an optimal evaluation strategy
\cite{AntoyEchahedHanus00JACM}---a conservative extension
of lazy functional programming and logic programming.

Curry is strongly typed so that
a Curry program consists of data type definitions
(introducing \emph{constructors} for data types) and
\emph{functions} or \emph{operations} on these types.
As an example, we show the definition of two
operations on lists: the well-known list concatenation \ccode{++}
and an operation \code{dup} which returns some list element having at least
two occurrences:\footnote{Note that Curry requires the explicit declaration
of free variables, as \code{x} in the rule of \code{dup},
to ensure checkable redundancy.
Anonymous variables, denoted by an underscore, need not be declared.}
\begin{curry}
(++) :: [a] -> [a] -> [a]       dup :: [a] -> a
[]     ++ ys = ys                 dup xs | xs == _$\,$++$\,$[x]$\,$++$\,$_$\,$++$\,$[x]$\,$++$\,$_
(x:xs) ++ ys = x : (xs ++ ys)            = x    where$\;$x$\;$free
\end{curry}
Since \code{dup} might deliver more than one result for a given argument,
e.g., \code{dup$\,$[1,2,2,1]} yields \code{1} and \code{2},
it is also called a \emph{non-deterministic operation}.
Such operations, which are interpreted as mappings
from values into sets of values \cite{GonzalezEtAl99},
are an important feature of contemporary functional logic languages.
To express failing computations, there is also a predefined operation
\code{failed} which has no value.

\begin{figure*}[t]
\begin{displaymath}
\begin{array}{lcl@{\hspace{5ex}}l}
P & ::= & D_1 \ldots D_m  & \mbox{(program)} \\
D & ::= & f(x_1,\ldots,x_n) = e  & \mbox{(function definition)} \\
e & ::= & x & \mbox{(variable) } \\
  & | & c(x_1,\ldots,x_n) & \mbox{(constructor application) } \\
  & | & f(x_1,\ldots,x_n)  & \mbox{(function call) } \\
  & | & e_1~\mathit{or}~e_2 & \mbox{(disjunction) } \\
  & | & \mathit{let}~x_1,\ldots,x_n ~\mathit{free~in}~ e
       & \mbox{(free variables) } \\
  & | & \mathit{let}~x = e ~\mathit{in}~ e'
       & \mbox{(let binding) } \\
  & | & \mathit{case}~x~\mathit{of}~\{p_1\to e_1; \ldots; p_n \to e_n\}
                         & \mbox{(case expression) } \\
p & ::= & c(x_1,\ldots,x_n)     & \mbox{(pattern)} 
\end{array}
\end{displaymath}
\caption{Syntax of the intermediate language FlatCurry}\label{fig:flatcurry}
\end{figure*}

Curry has more features than described so far.\footnote{%
Conceptually, Curry is intended as an extension of Haskell
although not all extensions of Haskell are actually supported.}
Due to these numerous features,
language processing tools for Curry (compilers, analyzers,\ldots)
often use an intermediate language where the
syntactic sugar of the source language has been eliminated
and the pattern matching strategy is explicit.
This intermediate language, called FlatCurry,
has also been used, apart from compilers,
to specify the operational semantics
of Curry programs \cite{AlbertHanusHuchOliverVidal05}
or to implement a modular framework for the analysis of
Curry programs \cite{HanusSkrlac14}.
Since we will use FlatCurry to describe and implement our
inference method,
we sketch the structure of FlatCurry programs.

The abstract syntax of FlatCurry is summarized
in Fig.~\ref{fig:flatcurry}.
A FlatCurry program consists of a sequence of function definitions
(for the sake of simplicity, data type definitions are omitted),
where each function is defined by a single rule.
Patterns in source programs are compiled into case expressions,
overlapping rules are joined by explicit disjunctions, and
arguments of constructor and function calls are variables
(introduced in left-hand sides, let expressions, or patterns).
We will write $\Fc$ for the set of defined operations and
$\Cc$ for the set of constructors of a program.
In order to provide a simple definition of our inference method,
we assume that FlatCurry programs satisfy the following properties:
\begin{itemize}
\item All variables introduced in a rule (parameters, free variables,
      let bindings, pattern variables) have unique identifiers.
\item For the sake of simplicity, let bindings are non-recursive,
      i.e., all recursion is introduced by functions
      (although our implemented tool supports recursive bindings).
\item The patterns in each case expression are non-overlapping
      and cover all data constructors of the type of the
      discriminating variable. Hence, if this type
      contains $n$ constructors, there are $n$ branches without
      overlapping patterns. This can be ensured by adding
      missing branches with failure expressions (\code{failed}).
\end{itemize}
Usually, the front end of a Curry compiler transforms
source programs into such a form for easier compilation
\cite{Antoy01PPDP,AntoyHanusJostLibby20}.
For instance, the operation \code{head} is transformed into
the following FlatCurry program:
\begin{curry}
head(zs) = case zs of { x:xs -> x ; [] -> failed }
\end{curry}

\section{Call Types and Abstract Types}
\label{sec:atypes}

We consider a computation as \emph{non-failing} if it does not stop
due to a pattern mismatch or a call to \code{failed}.
In order to infer conditions on arguments of operations so that
the evaluation of an operation does not fail,
we will analyze the rules of each operation.\footnote{Note that
we do not consider external failures of operations,
like file access errors, since they need to be handled differently.}
For instance, the operation \code{head} is defined by the single rule
\begin{curry}
head (x:xs) = x
\end{curry}
Since there is no rule covering the empty list,
the condition for a non-failing evaluation of \code{head} is the
non-emptiness of the argument list.
Sometimes the exact condition requires more advanced descriptions.
For instance, consider the operation \code{lastTrue} defined by
\begin{curry}
lastTrue [True]   = True
lastTrue (x:y:ys) = lastTrue (y:ys)
\end{curry}
The evaluation of a call \code{lastTrue$\;l$} does not fail
if the argument list $l$ ends with \code{True}.
Although such lists could be finitely described using
regular types \cite{DartZobel92},
such a description is impossible for arbitrary operations.
For instance, if some branch in a condition of an operation
causes a failure but the condition of the branch contains a function call,
the failure is only relevant if the function call terminates.
Due to the undecidability of the halting problem, we cannot hope
to infer exact non-failure conditions.

Due to this general problem, we \emph{approximate}
non-failure conditions so that the evaluation of a call where the arguments
satisfy the non-failure conditions is non-failing.
However, there might be successfully evaluable calls which do not satisfy
the inferred non-failure condition.

In order to support different structures to approximate
non-failure conditions, we do not fix a language for call types
but assume that there is a domain $\Ac$ of \emph{abstract types}.
Elements of this domain describe sets of concrete \emph{data terms},
i.e., terms consisting of data constructors only.
There are various options for such abstract types,
like depth-$k$ abstractions \cite{SatoTamaki84} or
regular types \cite{DartZobel92}.
The latter have been used to infer success types to analyze
logic programs \cite{GallagherHenriksen04}, whereas depth-$k$ abstractions
were used in the abstract diagnosis of functional programs
\cite{AlpuenteEtAl02} or in the abstraction of term rewriting systems
\cite{BertEchahed95,BertEchahedOstvold93}.
Since regular types are more complex and computationally more expensive,
we use depth-$k$ abstractions in our examples.
In this domain, denoted by $A_k$, subterms exceeding the given depth $k$
are replaced by a specific constant ($\top$) that represents any term.
Since the size of this domain is quickly growing for $k>1$, we use $k=1$
in examples, i.e.,
terms are approximated by their top-level constructors.
As we will see, this is often sufficient in practice to obtain
reasonable results.
Nevertheless, our technique and implementation is parametric
over the abstract type domain (results w.r.t.\ different domains
are shown in Appendix~\ref{appendix:domains}).

If $\Cc$ is the set of data constructors,
\emph{depth-1 types} can be simply described by the set
\[
\Ac_1 = \{ D \subseteq \Cc \mid
           \mbox{all constructors of $D$ belong to the same type} \}
        \cup \{ \top \}
\]
Hence, each element of $\Ac_1$ is either a set of data constructors of the
same type or $\top$.
The latter denotes the set of all data terms when no type information
is available.

Following the framework of abstract interpretation \cite{CousotCousot77},
the meaning of abstract values is specified by a concretization function
$\gamma$. For $\Ac_1$, $\gamma$ is defined by
\[
\begin{array}{r@{~~}c@{~~}l}
\gamma(\top) & = & \{ t \mid t \mbox{~is a data term}\}\\
\gamma(D)    & = & \{ t \mid t = c(t_1,\ldots,t_n) \mbox{~is a data term with }
                             c \in D \}
\end{array}
\]
Thus, $\emptyset$ is the bottom element of this domain
w.r.t.\ the standard ordering defined by
\[
\begin{array}{r@{~\sqsubseteq~}l@{\quad}l}
a & \top & \mbox{for any } a \\
a_1 & a_2 & \mbox{if } a_1 \subseteq a_2
\end{array}
\]
In the following, we present a framework for the inference of
call types which is parametric over the abstract domain $\Ac$.
Thus, we assume that $\Ac$ is a lattice with an ordering $\sqsubseteq$,
greatest lower bound ($\sqcap$) and least upper bound ($\sqcup$) operations,
a least or bottom element $\bot$, and a greatest or top element $\top$.
Furthermore, for each $n$-ary data constructor $c$,
there is an \emph{abstract constructor application}
$c^{\alpha}$ which maps abstract values $a_1,\ldots,a_n$
into an abstract value $a$
such that $c(t_1,\ldots,t_n) \in \gamma(a)$ for all
$t_1 \in \gamma(a_1),\ldots,t_n \in \gamma(a_n)$.
For the domain $\Ac_1$, which we use in the following in concrete examples,
this can be defined by $c^{\alpha}(x_1,\ldots,x_n) = \{c\}$
(it could also be defined by $c^{\alpha}(x_1,\ldots,x_n) = \top$
but this yields less precise approximations).

We use $\Ac$ to specify \emph{call types}
or \emph{non-failure conditions} for operations.
Let $f$ be a unary operation (the extension to more than one argument
is straightforward).
A call type $C \in \Ac$ is \emph{correct} for $f$
if the evaluation of $f(t)$ is non-failing for any $t \in \gamma(C)$.
For instance, the depth-1 type $\{\code{:}\}$ is correct for the
operations \code{head} or \code{tail} defined above.

In order to verify the correctness of call types for a program,
we have to check whether each call of an operation satisfies its call type.
Since this requires the analysis of conditions and other operations
(see the operation \code{readCmd} defined in Sect.~\ref{sec:readCmd}),
we will approximate the input/output behavior of operations,
as described in the following section.

\section{In/Out Types}
\label{sec:inouttypes}

To provide a fully automatic inference method for call types,
we need some knowledge about the behavior of auxiliary operations.
For instance, consider the operation \code{null} defined by
\begin{curry}
null :: [a] -> Bool
null []     = True
null (x:xs) = False
\end{curry}
This operation is used in the definition of
\code{readCmd} (see Sect.~\ref{sec:readCmd})
to ensure that \code{head} and \code{tail} are applied to
non-empty lists.
In order to verify this property, we have to infer that,
if \code{null\;ws} evaluates to \code{False},
the argument is a non-empty list.

For this purpose, we associate an in/out type to each operation.
An \emph{in/out type} $\io$ for an $n$-ary operation $f$
is a set of elements containing a sequence of $n+1$ abstract types, i.e.,
\[
\io \subseteq \{ a_1 \cdots a_n \iomapsto a \mid a_1,\ldots,a_n,a \in \Ac \}
\]
The first $n$ components of each element approximate input values
(where we write $\epsilon$ if $n=0$)
and the last component approximate output values associated to the inputs.
An in/out type $\io$ is \emph{correct} for $f$
if, for each value $t'$ of $f(t_1,\ldots,t_n)$,
there is some $a_1 \cdots a_n \iomapsto a \in \io$
such that $t_i \in \gamma(a_i)$ ($i = 1,\ldots,n$) and $t' \in \gamma(a)$.

Thus, in/out types are disjunctions of possible input/output behaviors
of an operation.
For instance, a correct in/out type of \code{null} w.r.t.\ $\Ac_1$ is
\begin{curry}
$\{ \{\code{[]}\} \iomapsto \{\code{True}\}, \{\code{:}\} \iomapsto \{\code{False}\} \}$
\end{curry}
Another trivial and less precise in/out type is
$\{ \top \iomapsto \top \}$.

In/out types allow also to express non-terminating operations.
For instance, a correct in/out type for the operation \code{loop}
defined by
\begin{curry}
loop = loop
\end{curry}
is $\{ \epsilon \iomapsto \emptyset \}$. The empty type in the result indicates
that this operation does not yield any value.

Similarly to call types, we approximate in/out types
since the inference of precise in/out types is intractable in general.
For this purpose, we analyze the definition of each operation
and associate patterns to result values.
Result values are based on general information about the
abstract result types of operations.
Therefore, we assume that there is a mapping
$
R : \Fc \to \Ac
$
which associates to each defined function $f \in \Fc$
an abstract type $R(f) \in \Ac$ approximating the possible values
to which $f$ (applied to some arguments) can be evaluated.
For instance,
$R(\code{loop}) = \emptyset$,
$R(\code{null}) = \{\code{False},\code{True}\}$, and
$R(\code{head}) = \top$
(w.r.t.\ the domain $\Ac_1$).
Approximations for $R$ can be computed in a straightforward
way by a fixpoint computation.
Using the Curry analysis framework CASS \cite{HanusSkrlac14},
this program analysis can be defined in 20 lines of code---basically
a case distinction on the structure
of FlatCurry programs.\footnote{See module \code{Analysis.Values}
of the Curry package
\url{https://cpm.curry-lang.org/pkgs/cass-analysis.html}}

\begin{figure*}[t]
\begin{displaymath}
\begin{array}{r@{\hspace{2ex}}cl}
\rulename{Var} &
  \Gamma ~\vdash x : \{ \Gamma \iomapsto \Gamma(x) \}
  & \mbox{\hspace*{-0ex}($x$ variable)} \\[1ex]
\rulename{Cons} &
  \Gamma ~\vdash c(x_1,\ldots,x_n) : \{ \Gamma \iomapsto c^{\alpha}(\Gamma(x_1),\ldots,\Gamma(x_n)) \}
  & \mbox{\hspace*{-0ex}($c$ constructor)} \\[1ex]
\rulename{Func} &
  \Gamma ~\vdash f(x_1,\ldots,x_n) : \{ \Gamma \iomapsto R(f) \}
  & \mbox{\hspace*{-0ex}($f$ operation)} \\[1ex]
\rulename{Or} &
 {\displaystyle
  \frac{\Gamma ~\vdash e_1 : io_1 \quad \Gamma ~\vdash e_2 : io_2}
       {\Gamma ~\vdash e_1~or~e_2 : io_1 \cup io_2}
} & \\[3ex]

\rulename{Free} & {\displaystyle
  \frac{\Gamma[\seq{x_n \mapsto \top}] \vdash e : io }
{\Gamma \vdash  \mathit{let}~x_1,\ldots,x_n ~\mathit{free~in}~ e : io }
} & \\[3ex]

\rulename{Let} & {\displaystyle
  \frac{\Gamma[x \mapsto \top] \vdash e' : io }
       {\Gamma \vdash let~ x = e~ in~e' : io } } & \\[3ex]

\rulename{Case} &
{\displaystyle
  \frac{\Gamma_1 \vdash e_1 : io_1 ~~\ldots~~ \Gamma_n \vdash e_n : io_n}
       {\Gamma \vdash \mathit{case}~x~\mathit{of}~\{p_1\to e_1; \ldots; p_n \to e_n\} : io_1\cup\ldots\cup io_n} } \\[3ex]
 & \mbox{ where } p_i = c_i(\seq{x_{n_i}}) \mbox{ and }
   \Gamma_i = \Gamma[x \mapsto c_i^{\alpha}(\seq{\top}),
   \seq{x_{n_i} \mapsto \top}]
\end{array}
\end{displaymath}
\caption{Approximation of in/out types}\label{fig:inouttypes}
\end{figure*}

Our actual approximation of in/out types is defined
by the rules in Fig.~\ref{fig:inouttypes}.
A sequence $o_1,\ldots,o_n$ of objects is abbreviated by $\seq{o_n}$.
We use a \emph{type environment} $\Gamma$ which maps variables
into abstract types.
We denote by $\Gamma[x \mapsto e]$ the environment $\Gamma'$ with
$\Gamma'(x) = e$ and $\Gamma'(y) = \Gamma(y)$ for all $x \neq y$.
The judgement
$
\Gamma \vdash e : \{ \seq{\Gamma_k \iomapsto a_k} \}
$
is interpreted as ``the evaluation of the expression $e$
in the context $\Gamma$ yields a new context $\Gamma_i$ and result value
of abstract type $a_i$, for some $i \in \{1,\ldots,k\}$.''
To infer an in/out type $io$ of an operation $f$ defined by
$f(x_1,\ldots,x_n) = e$,
we derive the judgement
$\{ \seq{x_n \mapsto \top} \} \vdash e : \{ \seq{\Gamma_k \iomapsto a_k} \}$
and return the in/out type
\begin{curry}
$io = \{ \Gamma_i(x_1) \cdots \Gamma_i(x_n) \iomapsto a_i \mid i=1,\ldots,k \}$
\end{curry}
Thus, we derive an in/out type without any restriction on the arguments.

Let us consider the inference rules in more detail.
In the case of variables or applications, the type environment is not changed
and the approximated result is returned,
e.g., the abstract type of the variable (rule \rulename{Var}),
the abstract representation of the constructor (rule \rulename{Cons}),
or the approximated result value of the operation (rule \rulename{Func}).
Rule \rulename{Or} combines the results of the different branches.
Rules \rulename{Free} and \rulename{Let} add the new variables
to the type environment with most general types.
Although one could refine these types,
we try to keep the analysis simple since this seems
to be sufficient in practice.

The most interesting rule is \rulename{Case}.
The results from the different branches are combined,
but inside each branch, the type of the discriminating variable $x$
is refined to the constructor of the branch.
For instance, consider the operation
\begin{curry}
null(zs) = case zs of { [] -> True ; (x:xs) -> False }
\end{curry}
If we analyze the in/out type with our rules, we start with the
type environment $\Gamma_0 = \{ \code{zs} \mapsto \top \}$.
Inside the branch, $\Gamma_0$ is refined to
$\Gamma_1 = \{ \code{zs} \mapsto \{\code{[]}\} \}$ and
$\Gamma_2 = \{ \code{zs} \mapsto \{\code{:}\}, \code{x} \mapsto \top,
               \code{xs} \mapsto \top \}$, respectively,
so that the in/out type (w.r.t.\ $\Ac_1$) derived for \code{null} is
$\{ \{\code{[]}\}\iomapsto \{\code{True}\}, \{\code{:}\}\iomapsto \{\code{False}\} \}$.

In our implementation, we keep in/out types in a normalized form
where different pairs with identical in types are joined
by the least upper bound of their out types.
Moreover, the in/out types of failed branches are omitted so that we obtain
\begin{curry}
head : $\{ \{\code{:}\}\iomapsto \top \}$
tail : $\{ \{\code{:}\}\iomapsto \top \}$
\end{curry}

\section{Inference and Checking of Call Types}
\label{sec:inference}

Based on the pieces introduced in the previous sections,
we can present our method to infer and verify call types
for all operations in a given program.
Basically, our method performs the following steps:
\begin{enumerate}
\item
The in/out types for all operations are computed
(see Sect.~\ref{sec:inouttypes}).
\item
Initial call types for all operations are computed
by considering the left-hand sides or case structure of their
defining rules.
\item
These call types are abstracted w.r.t.\ the abstract type domain.
\item
For each call to an operation $g$ occurring in a rule defining operation $f$,
we check, by considering the call structure and in/out types,
whether the call type of $g$ is satisfied.
\item
If some operation cannot be verified due to unsatisfied
call type restrictions, its call type is refined
by considering the additional call-type constraints
due to operations called in its right-hand side,
and start again with step 4.
\end{enumerate}
This fixpoint computation terminates if the abstract type domain
is finite (which is the case for depth-$k$ types)
or it is ensured that there are only finitely many refinements
for each call type in step 5 (which could be ensured
by widening steps in infinite abstract domains \cite{CousotCousot77}).
In the worst case, an empty call type might be inferred for some
operation. This does not mean that this operation is not useful
but one has to encapsulate its use with some safeness check.

In the following, we describe these steps in more detail.

\subsection{Initial Call Types}

Concrete call types are easy to derive by considering
the structure of \code{case} expressions in the transformed FlatCurry program.
If all constructors of some data type are covered in non-failed branches
of some \code{case} construct, there is no call type restriction
due to this pattern matching. Otherwise, the call type restriction
consists of those constructors occurring in non-failed branches.
For instance, the operation \code{null} has no call type restriction,
whereas the operations \code{head} and \code{tail} have
failed branches for the empty list so that the call type restriction
could be expressed by the set of terms
\[
\{ t_1 \code{:} t_2 \mid t_1,t_2 \mbox{ are arbitrary terms} \}
\]
As already discussed, we map such sets into a finite representation
by using abstract types.
Hence, the \emph{abstract call type} of an $n$-ary operation
is a sequence of elements of $\Ac$ of length $n$.
We say that such a type is \emph{trivial} if all elements in
this sequence are $\top$.
In case of the abstract type domain $\Ac_1$,
the set above is abstracted to $\{\code{:}\}$, thus, it is non-trivial.
Since the derivation of concrete call types and their abstraction
is straightforward, we omit further details here.

\subsection{Call Type Checking}

We assume that two kinds of information are given
for a defined operation $f$:
\begin{itemize}
\item An in/out type $\IO(f)$ approximating the input/output behavior of $f$.
\item An abstract call type $\CT(f)$ specifying the requirements
      to evaluate $f$ without failure.
\end{itemize}
$\IO(f)$ can be computed as shown in Sect.~\ref{sec:inouttypes}.
$\CT(f)$ can be approximated as discussed above,
but we have to show that all calls to $f$ actually satisfy these
requirements.
This is the purpose of the inference system shown in
Fig.~\ref{fig:calltypecheck}.

As discussed in Sect.~\ref{sec:inouttypes},
it is important to have information about the input/output behavior
of operations. Therefore, we introduced the notion of in/out types.
Now we use this information to approximate values of variables
occurring in program rules and pass this information through
the rules during checking time.
For this purpose, we use \emph{variable types} which are triples
of the form $(z,\io,x_1 \ldots x_n)$
where $z,x_1,\ldots,x_n$ are program variables and $\io$ is an in/out type
for an $n$-ary operation.
This is interpreted as: $z$ might have some value
of the result type $a$ for some $a_1 \ldots a_n \iomapsto a \in \io$
and, in this case, $x_1,\ldots,x_n$ have values of type $a_1,\ldots,a_n$,
respectively.
For instance, the variable type
\begin{curry}
$(z,~\{ \{\code{[]}\}\iomapsto \{\code{True}\}, \{\code{:}\}\iomapsto \{\code{False}\} \},~ xs)$
\end{curry}
expresses that $z$ might have value \code{True} and $xs$ is an empty list,
or $z$ has value \code{False} and $xs$ is a non-empty list.
Since we approximate values, we abstract a set of variable environments
with concrete values for variables to a set of variable types.
If such a set contains only one triple for some variable
and the $io$ component is a one-element set, we can use it for
definite reasoning.
To have a more compact notation for the abstract type of a program variable,
we denote by $x::a$ the triple $(x, \{ \epsilon \mapsto a \}, \epsilon)$.

Now we have a closer look at the rules of Fig.~\ref{fig:calltypecheck}.
This inference system derives judgements of the form
$\Delta, z = e \vdash \Delta'$
containing sets of variable types $\Delta, \Delta'$,
a variable $z$, and an expression $e$.
This is interpreted as
``if $\Delta$ holds, then the expression $e$ evaluates without a failure
and, if $z$ is bound to the result of this evaluation, $\Delta'$ holds.''
\label{sec:func-call-type-checking}
To check the call type $a_1 \ldots a_n$ of an operation $f$ defined by
$f(x_1,\ldots,x_n) = e$,
we try to derive the judgement
\[
\{ x_1 :: a_1, \ldots, x_n :: a_n \}, z = e \vdash \Delta
\]
for some fresh variable $z$.
Thus, we assign the call types as initial values of the parameters
and analyze the right-hand side of the operation.

\begin{figure*}[t]
\begin{displaymath}
\begin{array}{r@{\hspace{2ex}}cl}
\vrulename{Var} &
  \Delta, z = x \vdash \{ (z, \{ \epsilon \mapsto \Delta(x) \}, \epsilon) \}
  & \\[1ex]
\vrulename{Cons} &
  \Delta, z = c(x_1,\ldots,x_n) \vdash
  \{\iosubst{z}{\{\top^n \iomapsto c^{\alpha}(\seq{\Delta(x_n)})\}}{x_1 \ldots x_n}\}
  & \\[1ex]
\vrulename{Func} &
  {\displaystyle
  \frac
    {\CT(f) = a_1 \ldots a_n \quad \Delta(x_i) \sqsubseteq a_i ~(i=1,\ldots,n)}
    {\Delta,z = f(x_1,\ldots,x_n) \vdash \{\iosubst{z}{\IO(f)}{x_1 \ldots x_n} \}} }
  & \\[3ex]
\vrulename{Or} &
 {\displaystyle
  \frac{\Delta,z = e_1 \vdash \Delta_1 \quad \Delta,z = e_2 \vdash \Delta_2}
       {\Delta,z = e_1~or~e_2 \vdash \Delta_1 \cup \Delta_2 }
} & \\[3ex]

\vrulename{Free} & {\displaystyle
 \frac{\Delta \cup \{ x_1 :: \top,\ldots, x_n :: \top \}, z = e \vdash \Delta' }
  {\Delta, z = \mathit{let}~x_1,\ldots,x_n ~\mathit{free~in}~ e \vdash \Delta' }
} & \\[3ex]

\vrulename{Let} & {\displaystyle
  \frac{\Delta, x = e \vdash \Delta' \quad
        \Delta \cup \Delta', z = e' \vdash \Delta'' }
{\Delta, z = let~ x = e~ in~e' \vdash \Delta'' }
} \\[3ex]

\vrulename{Case} &
{\displaystyle
  \frac{\Delta_{r_1}, z = e_{r_1} \vdash \Delta'_{r_1} ~~\ldots~~ \Delta_{r_k}, z = e_{r_k} \vdash \Delta'_{r_k}}
       {\Delta, z = \mathit{case}~x~\mathit{of}~\{p_1\to e_1; \ldots; p_n \to e_n\} \vdash \Delta'_{r_1} \cup\ldots\cup \Delta'_{r_k}} } \\[3ex]
 & \mbox{ where } p_i = c_i(\seq{x_{n_i}}),
   \Delta_i = (\Delta \wedge [x \mapsto c_i])
              \cup \{ x_1 :: \top, \ldots, x_{n_i} :: \top \}, \\
 & \mbox{ and } r_1,\ldots,r_k
   \mbox{ are the reachable branches (i.e., $\Delta_{r_j}(x) \neq \bot$)}\\

\end{array}
\end{displaymath}
\caption{Call type checking}\label{fig:calltypecheck}
\end{figure*}

Keeping the interpretation of variable types in mind,
the inference rules are not difficult to understand.
$\Delta(x)$ denotes the least upper bound of all abstract type
information about variable $x$ available in $\Delta$,
which is defined by
\[
\Delta(x) = \bigsqcup\; \{ a \mid
 (x,\{\ldots,a_1 \ldots a_n \iomapsto a,\ldots\},\ldots) \in \Delta \}
\]
Rule \vrulename{Var} is immediate since the evaluation of a value
cannot fail so that we set the result $z$ to the abstract type of $x$.
Rule \vrulename{Cons} adds the simple condition that $z$ is bound
to the constructor $c$ after the evaluation
($\top^n = \top \ldots \top$ is a sequence of $n$ $\top$ elements).
Rule \vrulename{Func} is the first interesting rule.
The condition states that the abstract arguments of the function
must be smaller than the required call type so that the
concrete values are in a subset relationship.
If the requirements on call types hold, the operation
is evaluable and we connect the results and the arguments
with the in/out type of the operation.
The rules for disjunctions and free variable introduction
are straightforward.
In rule \vrulename{Let}, the result of analyzing
the local binding is used to analyze the expression.
We finally discuss the most important rule for case selections.

In rule \vrulename{Case}, $\Delta \wedge [x \mapsto c_i]$
denotes the set of variable types $\Delta$ modified by the definite binding
of $x$ to the constructor $c_i$.
This means that, if $\Delta$ contains a triple $(x,io,xs)$,
all result values in $io$ which are incompatible to $c_i$ are removed.
After this modification of $\Delta$, it may happen that
$\Delta(x)$ is the empty type, i.e., there is no concrete
value which $x$ can have so that this branch is \emph{unreachable}.
Therefore, the right-hand side of this branch need not be analyzed
so that rule \vrulename{Case} does not consider them.
For the remaining reachable branches, the right-hand side is analyzed
with the modified set of variable types so that
the value in the specific branch value is considered.

As an example, we check the simple operation
\begin{curry}
f(x) = let y = null(x) in case y of True  -> True
                                    False -> head(x)
\end{curry}
Recall that, for the abstract type domain $\Ac_1$,
the in/out type of \code{null} is
\begin{curry}
$IO(\code{null}) = \{ \{\code{[]}\}\iomapsto \{\code{True}\}, \{\code{:}\}\iomapsto \{\code{False}\} \}$
\end{curry}
and the abstract call type of \code{head} is $\{\code{:}\}$.
When we check the right-hand side of the definition of \code{f},
we start the checking of the \code{case} (after having checked
the \code{let} binding) with the set of variable types
\begin{curry}
$\Delta_1 = \{ (\code{x}, \{\epsilon \iomapsto \top\}, \epsilon),~ (\code{y}, \{ \{\code{[]}\}\iomapsto \{\code{True}\}, \{\code{:}\}\iomapsto \{\code{False}\} \}, \code{x}) \}$
\end{curry}
The check of the first case branch is immediate.
For the second case branch, we modify the previous set of variable types
to $\Delta_2 = \Delta_1 \wedge [\code{y} \mapsto \code{False}]$
so that we have
\begin{curry}
$\Delta_2 = \{ (\code{x}, \{\epsilon \iomapsto \top\}, \epsilon),~ (\code{y}, \{ \{\code{:}\}\iomapsto \{\code{False}\} \}, \code{x}) \}$
\end{curry}
The definite binding for \code{y} implies a definite binding for \code{x}
so that $\Delta_2$ is equivalent to
\begin{curry}
$\Delta_3 = \{ (\code{x}, \{\epsilon \iomapsto \{\code{:}\}\}, \epsilon),~ (\code{y}, \{ \{\code{:}\}\iomapsto \{\code{False}\} \}, \code{x}) \}$
\end{curry}
Hence, if we check the call \ccode{head(x)} w.r.t.\ $\Delta_3$,
the abstract argument type is $\Delta_3(\code{x}) = \{\code{:}\}$
so that the call type of \code{head} is satisfied.

As we have seen in this example, sets of variable types should
be kept in a simplified form in order to deduce most precise
type information. For instance, the definite bindings of variables,
like $(\code{y}, \{ \{\code{:}\}\iomapsto \{\code{False}\} \}, \code{x})$,
should be propagated to get a definitive binding for \code{x}.
Although this is not explicitly stated in the inference rules,
we assume that it is always done whenever sets of variable types are
modified.

\subsection{Iterated Call Type Checking}

Consider the operation
\begin{curry}
hd(x) = head(x)
\end{curry}
Applying the inference rules of Fig.~\ref{fig:calltypecheck}
is not successful: the initial abstract call type for
\code{hd} is $\top$ so that the call type requirement for \code{head}
is not satisfied.

In order to compute call types for all operations,
we try to refine the call type of \code{hd}.
For this purpose, we collect the requirements on variables
for unsatisfied call types during the check of an operation.
If such a required type is on some variable occurring in the
left-hand side of an operation, the call type of the operation
is restricted and the operation is checked again.
In case of the operation \code{hd},
the failure in the call \code{head(x)} leads to the requirement
that \code{x} must have the abstract type $\{\code{:}\}$
so that we check \code{hd} again but with this new call type---which
is now successful.

There are also cases where such a refinement is not possible.
For instance, consider the slightly modified example
\begin{curry}
hdfree(x) = let y free in head(y)
\end{curry}
Since the type restriction $\{\code{:}\}$ on variable \code{y}
can not be obtained by restricting the call type of \code{hdfree},
we assume the most restricted call type $\CT(\code{hdfree}) = \{\}$.
This means that any call to \code{hdfree} might fail so that
one has to encapsulate calls to \code{hdfree} with some safeness check.

This strategy leads to an iterated analysis of call types.
In each iteration, either all call types can be verified
or the call type of some operation becomes more restricted.
This iteration always terminates if one can ensure finitely many refinements
of call types (which is the case for depth-$k$ types).

For an efficient computation of this fixpoint computation,
it is reasonable to use call dependencies of operations
so that one has to re-check only the more restricted operations
and the operations that use them.
We have implemented this strategy in our tool and
obtained a good improvement compared to the initial
naive fixpoint computation.
For instance, the prelude of Curry (the base module containing
a lot of basic definitions for arithmetic, lists, type classes, etc)
contains 1262 operations (public and also auxiliary operations).
After the first iteration, the call types of 14 operations
are refined so that 17 operations are reanalyzed in the next iteration.
Altogether, the check of the prelude requires five iterations.

\subsection{Extensions}

Up to now, we presented the analysis of a kernel language.
Since application programs use more features,
we discuss in the following how to cover all features
occurring in Curry programs.

\subsubsection{Literals}
Programs might contain numbers or characters which are not introduced
by explicit data definitions.
Although there are conceptually infinitely many literals,
their handling is straightforward.
A literal can be treated as a $0$-ary constructor.
Since there are only finitely many literals in each program,
the abstract types for a given program are also finite.
For instance, consider the operation
\begin{curry}
k 0 = 'a'
k 1 = 'b'
\end{curry}
The call type of \code{k} inferred w.r.t.\ domain $\Ac_1$ is
$\CT(\code{k}) = \{ \code{0}, \code{1} \}$.
Similarly, the in/out type of \code{k} is
$\IO(\code{k}) = \{ \{\code{0}\} \iomapsto \{\code{'a'}\},
                    \{\code{1}\} \iomapsto \{\code{'b'}\} \}$.

\subsubsection{External operations}
Usually, externally defined primitive operations do not fail
so that they have trivial call types.
There are a few exceptions which are handled
by explicitly defined call types,
like the always failing operation \code{failed},
or arithmetic operations like division.

\subsubsection{Higher-order operations}
Since it is seldom that generic higher-order operations
have functional parameters with non-trivial call types,
we take a simple approach to check higher-order operations.
We assume that higher-order arguments have trivial call types
and check this property for each call to a higher-order operation.
Thus, a call like \ccode{map$\;$head$\;$[[1,2],[3,4]]} is considered
as potentially failing.
Our practical evaluation shows this assumption provides reasonable results
in practice.

\subsubsection{Encapsulation}

Failures might occur during run time, either due to operations with
complex non-failure conditions (e.g., arithmetic) or
due to the use of logic programming techniques
with search and failures.
In order to ensure an overall non-failing application
in the presence of possibly failing subcomputations,
the programmer has to encapsulate such subcomputations
and then analyze its outcome, e.g., branching on the result of
the encapsulation.
For this purpose, one can use an exception handler
(which represents a failing computation as an error value)
or some method to encapsulate non-deterministic search (e.g.,
\cite{AntoyHanus09,BrasselHanusHuch04JFLP,LopezSanchez04,Lux99FLOPS}).
For instance, the primitive operation \code{allValues} returns all the values
of its argument expression in a list so that a failure corresponds to
an empty list.
In order to include such a primitive in our framework,
we simply skip the analysis of its arguments.
For instance, a source expression like \ccode{allValues$\;$(head$\;$ys)}
is not transformed it into
\ccode{let x = head$\;$ys in allValues$\;$x} (where \code{x} is fresh),
but it is kept as it is. Furthermore, rule \vrulename{Func}
is specialized for \code{allValues} so that the condition
on the arguments w.r.t.\ the call type is omitted
and the in/out type is trivial,
i.e., $\IO(\code{allValues}) = \{ \top \iomapsto \top \}$.
In a similar way, other methods to encapsulate
possibly non-deterministic and failing operations,
like \emph{set functions} \cite{AntoyHanus09}, can be handled.

\subsubsection{Errors as Failures}

The operation \code{error} is an external operation to emit an error
message and terminate the program (if it does not occur inside
an exception handler).
Since we are mainly interested to avoid internal programming errors,
\code{error} is not considered as a failing operation in the default mode.
Thus, if we change the definition of \code{head} into (as in Haskell)
\begin{curry}
head :: [a] -> a
head []     = error "head: empty list"
head (x:xs) = x
\end{curry}
the inferred call type is $\top$ so that the call \ccode{head$\;$[]}
is not considered as failing.
From some point of view, this is reasonable since the evaluation does not fail
but shows a result---the error message.

However, in safety-critical applications we want to be sure
that all errors are caught.
In this case, we can still use our framework and define the call type
of \code{error} as $\bot$ so that any call to \code{error} is considered
as failing.
Moreover, exception handlers can be treated similarly to encapsulated
search operators as described above.
In order to be flexible with the interpretation of \code{error},
our tool (see below) provides an option to set one of these two views
of \code{error}.

\section{Evaluation}
\label{sec:evaluation}

We have implemented the methods described above in a tool\footnote{%
Available as package
\url{https://cpm.curry-lang.org/pkgs/verify-non-fail-1.0.0.html}}
written in Curry.
In the following we evaluate it by discussing some examples
and applying it to various libraries.

First, we compare our approach to a previous tool
to verify non-failing Curry programs
\cite{Hanus18PPDP}.
In that tool the programmer has to annotate partially defined
operations with \emph{non-fail conditions}.
Based on these conditions, the tool extracts
proof obligations from a program which are sent to an SMT solver.
For instance, consider the operation to compute the last element
of a non-empty list:
\begin{curry}
last [x]      = x
last (_:x:xs) = last (x:xs)
\end{curry}
The condition to express the non-failure of this expression
must be explicitly defined as a predicate on the argument:
\begin{curry}
last'nonfail xs = not (null xs)
\end{curry}
This predicate together with the definition of the involved operations
are translated to SMT formulas and then checked by an SMT solver,
e.g., Z3 \cite{deMouraBjorner08}.
Using our approach, the abstract call type
$\CT(\code{last}) = \{\code{:}\}$ is automatically inferred and
the definition of \code{last} is successfully checked.
Actually, we tested our tool on various libraries and could
deduce almost all manually written non-fail conditions of \cite{Hanus18PPDP}.
Only in four prelude operations, our tool could not infer
these non-fail conditions since they contain arithmetic conditions on integers.
We leave it for future work to combine our approach with an SMT solver
to enable also successful checks in these cases.

Another interesting example is the operation \code{split}
from the library \code{Data.List}.
This operation takes a predicate and a list as arguments
and splits this list into sublists at positions where the predicate
holds. It is defined in Curry as
\begin{curry}
split :: (a -> Bool) -> [a] -> [[a]]
split _ []                 = [[]]
split p (x:xs) | p x       = [] : split p xs
               | otherwise = let (ys:yss) = split p xs
                             in (x:ys):yss
\end{curry}
The interesting point in this example is the pattern matching
in the right-hand side \ccode{let$\;$(ys:yss)$\;$=$\;\cdots$}.
In order to implement this pattern matching in a lazy manner,
specific selector operations are introduced
when this definition is transformed into a kernel language
like FlatCurry, since FlatCurry allows only variable bindings
but not constructor patterns in \code{let} expressions.
Thus, the FlatCurry code generated for this definition introduces
two selector operations (named \code{split\us{}ys} and \code{split\us{}yss}
below) to implement the lazy pattern matching in the \code{let} expression.
The standard front end of Curry implementations translates
the definition above into the following FlatCurry code
(which is bit more relaxed than required in Fig.~\ref{fig:flatcurry}):
\begin{curry}
split :: (a -> Bool) -> [a] -> [[a]]
split p zs = case zs of
  []     -> [] : []
  x : xs -> let px = apply p x
             in case px of
                  True  -> [] : (split p xs)
                  False -> let o = otherwise
                            in case o of
                                 True -> let ts  = split p xs
                                          in let ys  = split_ys ts
                                             in let yss = split_yss ts
                                                in (x : ys) : yss
                                 False -> failed

split_ys :: [[a]] -> [a]
split_ys zs = case zs of x : xs -> x

split_yss :: [[a]] -> [[a]]
split_yss zs = case zs of x : xs -> xs
\end{curry}
\code{apply} is a predefined primitive operation
to implement higher-order application.
The predefined operation \code{otherwise} is equivalent to \code{True}
so that the occurrence of \code{failed} is not reachable.

Note that the selector operations \code{split\us{}ys} and \code{split\us{}yss}
are partially defined (they correspond to \code{head} and \code{tail}).
Since they are generated during the compilation process,
a non-fail condition cannot be explicitly defined in the source program
so that the tool described in \cite{Hanus18PPDP} could not verify
this definition of \code{split}.
As mentioned in \cite{Hanus18PPDP}, it is necessary
to include explicit calls to \code{head} and \code{tail}
instead of the pattern matching of \code{let}.
Moreover, the post-condition
\begin{curry}
split'post p xs ys = not (null ys)
\end{curry}
had to be added and proved by a contract checker \cite{Hanus20FI}
so that this information is used by the verifier to ensure
that the recursive call to \code{split} always returns a non-empty list.

With our method, such manual additions are not required
since the call types of the introduced selector operations
are automatically inferred together with the in/out type
\begin{curry}
$\IO(\code{split}) = \{ \top\cdot\{\code{[]}\} \iomapsto \{\code{:}\},~ \top\cdot\{\code{:} \} \iomapsto \{\code{:}\} \}$
\end{curry}
Thanks to this in/out type, the abstract result of the recursive
call to \code{split} is $\{\code{:}\}$ which matches the call types
required for the selector operations.
Thus, in contrast to \cite{Hanus18PPDP},
no manual annotations or code modifications are necessary
to check the non-failure of \code{split}.

If our tool is applied to a Curry module,
it infers the in/out types and the call types
of all operations defined in this module and then
checks all branches and calls whether they might be failing.
If this is the case, the call types are refined
and the problematic ones are reported to the user.
Then the user can decide to either accept the refined call types
or modify the program code to handle possible failures so that
the call type does not need a refinement.

\begin{table}[t]
\centering
\begin{tabular}{l*{5}{@{~~~}r@{/}l}@{~~~}cc}
\hline
Module
 & \multicolumn{2}{c}{operations}
 & \multicolumn{2}{c}{\begin{tabular}{@{}c@{}}in/out \\ types\end{tabular}}
 & \multicolumn{2}{c}{\begin{tabular}{@{}c@{}}initial \\[-0.5ex] call \\[-0.5ex] types\end{tabular}}
 & \multicolumn{2}{c}{\begin{tabular}{@{}c@{}}final \\[-0.5ex] call \\[-0.5ex] types\end{tabular}}
 & \multicolumn{2}{c}{\begin{tabular}{@{}c@{}}final \\ failing\end{tabular}}
 & \begin{tabular}{@{}c@{}}itera- \\ tions\end{tabular}
 & \begin{tabular}{@{}c@{}}verify \\ time\end{tabular} \\
\hline
\code{Prelude} & 862 & 1262 & 605 & 857 & 24 & 32 & 63 & 71 & 45 & 53 & 5 & 969 \\
\code{Data.Char} & 9 & 9 & 0 & 0 & 0 & 0 & 0 & 0 & 0 & 0 & 1 & 272 \\
\code{Data.Either} & 7 & 11 & 5 & 9 & 2 & 2 & 2 & 2 & 0 & 0 & 1 & 113 \\
\code{Data.List} & 49 & 87 & 39 & 73 & 7 & 15 & 8 & 16 & 1 & 1 & 2 & 290 \\
\code{Data.Maybe} & 8 & 9 & 7 & 8 & 0 & 0 & 0 & 0 & 0 & 0 & 1 & 113 \\
\code{Numeric} & 5 & 7 & 0 & 2 & 0 & 0 & 0 & 0 & 0 & 0 & 1 & 273 \\
\code{System.Console.GetOpt} & 6 & 47 & 5 & 41 & 0 & 0 & 0 & 0 & 0 & 0 & 1 & 287 \\
\code{System.IO} & 32 & 51 & 10 & 12 & 0 & 0 & 0 & 0 & 0 & 0 & 1 & 115 \\
\code{Text.Show} & 4 & 4 & 4 & 4 & 0 & 0 & 0 & 0 & 0 & 0 & 1 & 110 \\
\hline \\ 
\end{tabular}
\caption{Inference of call types for some standard libraries}\label{table:calltypes}
\end{table}

Table~\ref{table:calltypes} contains the results of checking various
Curry libraries with our tool.
The ``operations'' column contains the number of public (exported) operations
and the number of all operations defined in the module.
Similarly, the following three columns shows the information for
public and all operations.
The ``in/out types'' column shows the numbers of non-trivial in/out types.
The initial and final call types are the number of non-trivial call types
computed at the beginning and obtained after some iterations
(the number of iterations is shown in the next to last column).
The ``final failing'' column contains the number of operations
where an empty call type is inferred, i.e., there is no precise
information about the required call types.
The last column shows the verification time in milliseconds.\footnote{%
We measured the verification time on a Linux machine 
running Ubuntu 22.04 with an Intel Core i7-1165G7 (2.80GHz) processor
with eight cores.}

As one can see from this table, even quite complex modules, like the prelude,
have only a few operations with non-trivial call types that need to be checked.
Therefore, the effort to infer and check modules is limited.
The higher number of failing operations in the prelude are
the various arithmetic division operators and enumeration and parsing
operations where a precise call type cannot be inferred.

\section{Related Work}
\label{sec:related}

The exclusion of run-time failures at compile time is a
practically relevant but also challenging issue.
Therefore, there are many approaches targeting it
so that we can only discuss a few of them.
We concentrate on approaches for functional and logic programming,
although there are also many in the imperative world.
As mentioned in the introduction, the exclusion of
dereferencing null pointers is quite relevant there.
As an example from object-oriented programming,
the Eiffel compiler uses appropriate type declarations
and static analysis to ensure that pointer dereference failures
cannot occur in accepted programs \cite{Meyer17}.

In logic programming, there is no common definition of
``non-failing'' due to different interpretations of non-determinism.
Whereas we are interested to exclude any failure in a top-level
computation, other approaches, like \cite{BuenoEtAl04,DebrayEtAl97},
consider a predicate in a logic program as non-failing
if at least one answer is produced.
Similarly to our approach, type abstractions
are used to approximate non-failure properties,
but the concrete methods are different.

Another notion of failing programs in a
dynamically typed programming language is based on success types,
e.g., as used in Erlang \cite{LindahlSagonas06}.
Success types over-approximate possible uses of an operation
so that an empty success type indicates an operation
that never evaluates to some value.
Thus, success types can show definite failures,
whether we are interested in definite non-failures.

Strongly typed programming languages are a reasonable basis
to check run-time failures at compile time, since the type system
already ensures that some kind of failures cannot occur
(``well-typed programs do not go wrong'' \cite{Milner78}).
However, failures due to definitions with partial patterns
are not covered by a standard type system.
Therefore, Mitchell and Runciman developed
a checker for Haskell to verify the absence of pattern-match errors
due to incomplete patterns \cite{MitchellRunciman07,MitchellRunciman08}.
Their checker extracts and solves specific constraints
from pattern-based definitions.
Although these constraints have similarities to the
abstract type domain $\Ac_1$, our approach is generic
w.r.t.\ the abstract type domain so that it can also deal with
more powerful abstract types.

An approach to handle failures caused by restrictions on
number arguments is described in \cite{JhalaMajumdarRybalchenko11}.
It is based on generating (arithmetic) constraints
which are translated into an imperative program
such that the constraints are satisfiable iff the translated program is safe.
This enables the inference of complex conditions on numbers,
but pattern matching with algebraic data types and logic-oriented
subcomputations are not supported.

Another approach to ensure the absence of failures is to make
the type system stronger or more expressive in order to encode non-failing
conditions in the types.
For instance, operations in
dependently typed programming languages, such as
Coq \cite{BertotCasteran04}, Agda \cite{Norell08}, or Idris \cite{Brady13},
must be totally defined, i.e., terminating and non-failing.
Such languages have termination checkers but non-fail conditions
need to be explicitly encoded in the types.
For instance, the definition of the operation \code{head} in
Agda \cite{Norell08} requires, as an additional argument,
a proof that the argument list is not empty.
Thus, \code{head} could have the type signature
\begin{curry}
head : {A : Set} -> (xs : List A) -> is-empty xs == ff -> A
\end{curry}
Therefore, each use of \code{head} must provide,
as an additional argument, an explicit proof
for the non-emptiness of the argument list \code{xs}.
Type-checked Agda programs do not contain run-time failures
but programming in a dependently typed language
is more challenging since the programmer has to construct
non-failure proofs.

Refinement types, as used in
LiquidHaskell \cite{VazouSeidelJhala14,VazouEtAl14},
are another approach to encode non-failing conditions or
more general contracts on the type-level.
Refinement types extend standard types by a predicate
that restricts the set of allowed values.
For instance, the applications of \code{head} to the empty list
can be excluded by the following refinement type \cite{VazouSeidelJhala14}:
\begin{curry}
head :: {xs : [a] | 0 < len xs} -> a
\end{curry}
The correctness of refinement types is checked by an SMT solver
so that they are more expressive than our non-failure conditions.
On the other hand, refinement types must be explicitly added
by the programmer whereas our goal is to infer non-failure conditions
from a given program without specific annotations.
This allows the use of potentially failing operations
in encapsulated subcomputations, which is relevant
to use logic programming techniques.
This aspect is also the motivation for
the non-failure checking tool proposed in \cite{Hanus18PPDP}.
As already discussed in Sect.~\ref{sec:evaluation},
the advantage of our tool is the automatic inference
of non-failing conditions which supports an easier application
to larger programs.

\section{Conclusions}
\label{sec:concl}

In this paper we proposed a new technique and a fully automatic tool
to check declarative programs for the absence of failing computations.
In contrast to other approaches,
our approach does not require the explicit specification
of non-fail conditions but is able to infer them.
In order to provide flexibility with the structure of non-fail conditions,
our approach is generic w.r.t.\ a domain of abstract types
to describe non-fail conditions.
Since we developed our approach for Curry,
it is also applicable to purely functional or logic programs.
Due to the use of Curry, we do not need to abandon
all potentially failing operations.
Partially defined operations and failing evaluations
are still allowed in logic-oriented subcomputations provided
that they are encapsulated in order to control possible failures.

Although the inference of non-fail conditions is based
on a fixpoint iteration and might yield, in the worst case,
an empty (i.e., always failing) condition,
our practical evaluation showed that even larger programs
contain only a few operations with non-trivial non-fail conditions
which are inferred after a small number of iterations.
When a non-trivial non-fail condition is inferred for some operation,
the programmer can either modify the definition of this operation
(e.g., by adding missing case branches)
or control the invocation of this operation by checking its outcome
with some control operator.

For future work, we plan to extend our approach
to built-in types, like integers, and infer non-failure conditions
on such types, like non-negative or positive numbers,
and check them using SMT solvers.
Furthermore, it is interesting to see whether
other abstract domains, e.g., regular types,
yield more precise results in application programs.
Our first experiments with depth-$k$ domains
(see Appendix~\ref{appendix:domains}) showed that
a rather simple domain seems sufficient in practical programs.


\newpage
\appendix

\section{Inference of Call Types with Different Domains}
\label{appendix:domains}

Our inference of call types is parametric w.r.t.\ the domain of abstract types.
In the examples of the paper, a simple domain, where data terms are abstracted
to their top-level constructors, is used.
The actual implementation supports also depth-$k$ abstractions
\cite{SatoTamaki84} with $k>1$.
Although deeper abstractions could provide more precision,
we show by the analysis of standard libraries that this is seldom
in practical programs.
The following tables show the results for $k=1$ (which is identical
to Table~\ref{table:calltypes}), $k=2$, and $k=5$.
Apart from the execution times, the results are identical.

\paragraph{Inference of call types with the depth-1 domain:}

\begin{center}
\begin{tabular}{l*{5}{@{~~~}r@{/}l}@{~~~}cc}
\hline
Module
 & \multicolumn{2}{c}{operations}
 & \multicolumn{2}{c}{\begin{tabular}{@{}c@{}}in/out \\ types\end{tabular}}
 & \multicolumn{2}{c}{\begin{tabular}{@{}c@{}}initial \\[-0.5ex] call \\[-0.5ex] types\end{tabular}}
 & \multicolumn{2}{c}{\begin{tabular}{@{}c@{}}final \\[-0.5ex] call \\[-0.5ex] types\end{tabular}}
 & \multicolumn{2}{c}{\begin{tabular}{@{}c@{}}final \\ failing\end{tabular}}
 & \begin{tabular}{@{}c@{}}itera- \\ tions\end{tabular}
 & \begin{tabular}{@{}c@{}}verify \\ time\end{tabular} \\
\hline
\code{Prelude} & 862 & 1262 & 605 & 857 & 24 & 32 & 63 & 71 & 45 & 53 & 5 & 969 \\
\code{Data.Char} & 9 & 9 & 0 & 0 & 0 & 0 & 0 & 0 & 0 & 0 & 1 & 272 \\
\code{Data.Either} & 7 & 11 & 5 & 9 & 2 & 2 & 2 & 2 & 0 & 0 & 1 & 113 \\
\code{Data.List} & 49 & 87 & 39 & 73 & 7 & 15 & 8 & 16 & 1 & 1 & 2 & 290 \\
\code{Data.Maybe} & 8 & 9 & 7 & 8 & 0 & 0 & 0 & 0 & 0 & 0 & 1 & 113 \\
\code{Numeric} & 5 & 7 & 0 & 2 & 0 & 0 & 0 & 0 & 0 & 0 & 1 & 273 \\
\code{System.Console.GetOpt} & 6 & 47 & 5 & 41 & 0 & 0 & 0 & 0 & 0 & 0 & 1 & 287 \\
\code{System.IO} & 32 & 51 & 10 & 12 & 0 & 0 & 0 & 0 & 0 & 0 & 1 & 115 \\
\code{Text.Show} & 4 & 4 & 4 & 4 & 0 & 0 & 0 & 0 & 0 & 0 & 1 & 110 \\
\hline \\ 
\end{tabular}
\end{center}

\bigskip

\paragraph{Inference of call types with the depth-2 domain:}
\begin{center}
\begin{tabular}{l*{5}{@{~~~}r@{/}l}@{~~~}cc}
\hline
Module
 & \multicolumn{2}{c}{operations}
 & \multicolumn{2}{c}{\begin{tabular}{@{}c@{}}in/out \\ types\end{tabular}}
 & \multicolumn{2}{c}{\begin{tabular}{@{}c@{}}initial \\[-0.5ex] call \\[-0.5ex] types\end{tabular}}
 & \multicolumn{2}{c}{\begin{tabular}{@{}c@{}}final \\[-0.5ex] call \\[-0.5ex] types\end{tabular}}
 & \multicolumn{2}{c}{\begin{tabular}{@{}c@{}}final \\ failing\end{tabular}}
 & \begin{tabular}{@{}c@{}}itera- \\ tions\end{tabular}
 & \begin{tabular}{@{}c@{}}verify \\ time\end{tabular} \\
\hline
\code{Prelude} & 862 & 1262 & 605 & 857 & 24 & 32 & 63 & 71 & 45 & 53 & 5 & 986 \\
\code{Data.Char} & 9 & 9 & 0 & 0 & 0 & 0 & 0 & 0 & 0 & 0 & 1 & 381 \\
\code{Data.Either} & 7 & 11 & 5 & 9 & 2 & 2 & 2 & 2 & 0 & 0 & 1 & 118 \\
\code{Data.List} & 49 & 87 & 39 & 73 & 7 & 15 & 8 & 16 & 1 & 1 & 2 & 405 \\
\code{Data.Maybe} & 8 & 9 & 7 & 8 & 0 & 0 & 0 & 0 & 0 & 0 & 1 & 118 \\
\code{Numeric} & 5 & 7 & 0 & 2 & 0 & 0 & 0 & 0 & 0 & 0 & 1 & 393 \\
\code{System.Console.GetOpt} & 6 & 47 & 5 & 41 & 0 & 0 & 0 & 0 & 0 & 0 & 1 & 425 \\
\code{System.IO} & 32 & 51 & 10 & 12 & 0 & 0 & 0 & 0 & 0 & 0 & 1 & 126 \\
\code{Text.Show} & 4 & 4 & 4 & 4 & 0 & 0 & 0 & 0 & 0 & 0 & 1 & 117 \\
\hline \\ 
\end{tabular}
\end{center}

\bigskip
\bigskip

\paragraph{Inference of call types with the depth-5 domain:}
\begin{center}
\begin{tabular}{l*{5}{@{~~~}r@{/}l}@{~~~}cc}
\hline
Module
 & \multicolumn{2}{c}{operations}
 & \multicolumn{2}{c}{\begin{tabular}{@{}c@{}}in/out \\ types\end{tabular}}
 & \multicolumn{2}{c}{\begin{tabular}{@{}c@{}}initial \\[-0.5ex] call \\[-0.5ex] types\end{tabular}}
 & \multicolumn{2}{c}{\begin{tabular}{@{}c@{}}final \\[-0.5ex] call \\[-0.5ex] types\end{tabular}}
 & \multicolumn{2}{c}{\begin{tabular}{@{}c@{}}final \\ failing\end{tabular}}
 & \begin{tabular}{@{}c@{}}itera- \\ tions\end{tabular}
 & \begin{tabular}{@{}c@{}}verify \\ time\end{tabular} \\
\hline
\code{Prelude} & 862 & 1262 & 605 & 857 & 24 & 32 & 63 & 71 & 45 & 53 & 5 & 987 \\
\code{Data.Char} & 9 & 9 & 0 & 0 & 0 & 0 & 0 & 0 & 0 & 0 & 1 & 431 \\
\code{Data.Either} & 7 & 11 & 5 & 9 & 2 & 2 & 2 & 2 & 0 & 0 & 1 & 120 \\
\code{Data.List} & 49 & 87 & 39 & 73 & 7 & 15 & 8 & 16 & 1 & 1 & 2 & 438 \\
\code{Data.Maybe} & 8 & 9 & 7 & 8 & 0 & 0 & 0 & 0 & 0 & 0 & 1 & 119 \\
\code{Numeric} & 5 & 7 & 0 & 2 & 0 & 0 & 0 & 0 & 0 & 0 & 1 & 424 \\
\code{System.Console.GetOpt} & 6 & 47 & 5 & 41 & 0 & 0 & 0 & 0 & 0 & 0 & 1 & 462 \\
\code{System.IO} & 32 & 51 & 10 & 12 & 0 & 0 & 0 & 0 & 0 & 0 & 1 & 127 \\
\code{Text.Show} & 4 & 4 & 4 & 4 & 0 & 0 & 0 & 0 & 0 & 0 & 1 & 117 \\
\hline \\ 
\end{tabular}
\end{center}

\newpage

\section{Correctness}

In the following, we show that inference systems defined in the paper
correctly approximate the evaluation of functional logic programs.
Similarly to the inference systems,
we consider the evaluation of FlatCurry programs
(see Fig.~\ref{fig:flatcurry}).
Function definitions are considered as rewrite rules and
do not contain unbound variables, i.e., all variables
occurring in the right-hand side are either parameters or
introduced by $\mathit{let}$ bindings or $\mathit{case}$ patterns.
FlatCurry expressions are evaluated by term rewriting
with the following specific evaluation rules:
\begin{itemize}
\item
A disjunction $e_1~\mathit{or}~e_2$ is non-deterministically evaluated,
i.e., reduced to either $e_1$ or $e_2$.
Thus, our reduction relation is non-deterministic so that
all possible evaluations are considered.
\item
The evaluation of $\mathit{let}~x_1,\ldots,x_n ~\mathit{free~in}~ e'$
guesses values (data terms) $t_1,\ldots,t_n$
for the free variables $x_1,\ldots,x_n$ so that it reduces to
$\rho(e')$ where $\rho = \{ x_1 \mapsto t_1,\ldots, x_n \mapsto t_n \}$.
Actual implementations defer this guessing to the point where
a value of a variable is needed, which is called narrowing.
Since we do not consider a specific operational semantics,
we use this quite general (non-deterministic) guessing step.
\item
A let binding $\mathit{let}~ x = e~ \mathit{in}~e'$
is evaluated by evaluating $e$ to a data term $t$\footnote{
In order to model non-strict evaluation, one could add
partial terms, i.e., an undefined value, and the possibility
to reduce a term to an undefined value, as in the
rewriting logic CRWL \cite{GonzalezEtAl99}.
For the sake of simplicity, we omit the introduction of partial terms.
}
and evaluating $\rho(e')$ where $\rho = \{ x \mapsto t \}$.
\item
A case expression
$\mathit{case}~x~\mathit{of}~\{p_1\to e_1; \ldots; p_n \to e_n\}$
is evaluated by reducing it to $\rho(e_i)$ where $x$ matches the pattern $p_i$
and $\rho$ binds the variables of $p_i$.
Note that the case expression covers all constructors of a type
so that such a reduction step is always possible.
\end{itemize}
In the following, we denote by
$\var{e}$ the set of all unbound variables of an expression $e$,
i.e., all variables occurring in $e$ which are not introduced
in patterns or let expressions.

\subsection{Correctness of In/Out Types}

First we show that the inference rules in Fig.~\ref{fig:inouttypes}
correctly approximate the input-output behavior of operations.
For this purpose, we denote by
 $\dom{\Gamma}$ the domain of a type environment $\Gamma$
(a mapping from variables into abstract types).
We state the correctness of the inference rules as follows.

\begin{theorem}
\label{theo:inout-rules-are-correct}
Let $\var{e} \subseteq \dom{\Gamma}$,
$\Gamma \vdash e : \{ \seq{\Gamma_k \iomapsto a_k} \}$
be derivable by the inference rules in Fig.~\ref{fig:inouttypes},
and $\sigma$ be a substitution with
$\sigma(x) \in \gamma(\Gamma(x))$ for all $x \in \var{e}$.
If $\sigma(e)$ is reducible to some data term $t$,
then there is some $i \in \{1,\ldots,k\}$ with
$t \in \gamma(a_i)$ and
$\sigma(x) \in \gamma(\Gamma_i(x))$ for all $x \in \var{e}$.
\end{theorem}
\begin{proof}
By induction on the height of the proof tree to derive
$\Gamma \vdash e : \io$.
We assume that the theorem's preconditions hold, i.e.,
$\var{e} \subseteq \dom{\Gamma}$,
$\Gamma \vdash e : \io$ with $\io = \{ \seq{\Gamma_k \iomapsto a_k} \}$
is derivable with an proof tree of height $h$,
$\sigma$ is a substitution with
$\sigma(y) \in \gamma(\Gamma(y))$ for all $y \in \var{e}$,
and $\sigma(e)$ is reducible to some data term $t$.
We have to show that there is some $i \in \{1,\ldots,k\}$ with
\begin{itemize}
\item[(1)] $t \in \gamma(a_i)$ and
\item[(2)] $\sigma(y) \in \gamma(\Gamma_i(y))$ for all $y \in \var{e}$.
\end{itemize}
Induction base $h=1$:
We distinguish the different axioms, i.e.,
the rules \rulename{Var}, \rulename{Cons}, and \rulename{Func}:

\begin{description}
\item[\textit{Rule \rulename{Var} is applied:}]
Then $e = x$, $k=1$, $\Gamma_1 = \Gamma$, $a_1 = \Gamma(x)$,
and $t = \sigma(x)$ since $\sigma(x)$ is a data term.
Since $\sigma(x) \in \gamma(\Gamma(x))$, properties (1) and (2) hold.

\item[\textit{Rule \rulename{Cons} is applied:}]
In this case we have $e = c(x_1,\ldots,x_n)$,
$k=1$, $\Gamma_1 = \Gamma$, and
$a_1 = c^{\alpha}(\Gamma(x_1),\ldots,\Gamma(x_n))$.
$\sigma(x_i) \in \gamma(\Gamma(x_i)) = \gamma(\Gamma_1(x_i))$
($i=1,\ldots,n$) so that (2) holds.
Since $\sigma(e)$ is a data term,
$t = \sigma(e) = c(\sigma(x_1),\ldots,\sigma(x_n))$.
This implies property (1)
since $\sigma(x_i) \in \gamma(\Gamma(x_i))$ ($i=1,\ldots,n$) and
$t = c(\sigma(x_1),\ldots,\sigma(x_n)) \in \gamma(c^{\alpha}(\Gamma(x_1),\ldots,\Gamma(x_n)))$
by definition of the abstract constructor application
(see Sect.~\ref{sec:atypes}).

\item[\textit{Rule \rulename{Func} is applied:}]
Then $e = f(x_1,\ldots,x_n)$,
$k=1$, $\Gamma_1 = \Gamma$, and $a_1 = R(f)$.
Hence $\sigma(x_i) \in \gamma(\Gamma(x_i)) = \gamma(\Gamma_1(x_i))$
($i=1,\ldots,n$) so that (2) holds.
Furthermore, $\sigma(e) = f(t_1,\ldots,t_n)$ for some
terms $t_1,\ldots,t_n$.
Since $R(f)$ approximates all possible values of a call to $f$,
$t \in \gamma(R(f))$.
\end{description}
Induction step, i.e., $h>1$:
In this case, one of the rules \rulename{Or}, \rulename{Free},
\rulename{Let}, or \rulename{Case} is applied:

\begin{description}
\item[\textit{Rule \rulename{Or} is applied:}]
Then $e = e_1~\mathit{or}~e_2$ so that either $\sigma(e_1)$ or $\sigma(e_2)$
is reducible to $t$.
Consider the case that $\sigma(e_1)$ is reducible to $t$
(the other case can be similarly proved).
Since $\Gamma \vdash e_1 : \io_1$ with
$\io_1 = \{ \seq{\Gamma'_{k'} \iomapsto a'_{k'}} \}$
is derivable and $\sigma(e_1)$ is reducible to $t$,
by induction hypothesis,
there is some $j \in \{1,\ldots,k'\}$ with
$\sigma(x) \in \gamma(\Gamma'_j(x))$ for all $x \in \var{e_1}$
and $t \in \gamma(a'_j)$.
If there is some variable $y \in \var{e_2}$ and $y \not\in \var{e_1}$,
$\sigma(y) \in \gamma(\Gamma(y)) = \gamma(\Gamma'_j(y))$,
since an abstract value in a type environment is only changed
if this variable occurs in a case expression as a discriminator
(see the rules in Fig.~\ref{fig:inouttypes})
which cannot be the case for $y$ in expression $e_1$.
Since $\io_1 \subseteq \io$, properties (1) and (2) hold.

\item[\textit{Rule \rulename{Free} is applied:}]
A reduction of the expression
$e = \mathit{let}~x_1,\ldots,x_n ~\mathit{free~in}~ e'$ guesses values
for the free variables $x_1,\ldots,x_n$ in order to reduce $e'$.
Since $\sigma(e)$ reduces to $t$, there must be a substitution
$\rho = \{ \seq{x_n \mapsto t_n} \}$ for the free variables such that
$\sigma(e)$ is reduced to $\sigma(\rho(e'))$
which is reducible to $t$.
$\var{e'} \subseteq \dom{\Gamma[\seq{x_n \mapsto \top}]}$
since $\var{e} \subseteq \dom{\Gamma}$.
Hence, by the induction hypothesis,
there is some $i \in \{1,\ldots,k\}$ with
$\sigma(\rho(y)) \in \gamma(\Gamma_i(y))$ for all $y \in \var{e'}$
and $t \in \gamma(a_i)$.
Thus, property (1) holds.
Since $\rho$ is the identity on variables different from $x_1,\ldots,x_n$,
$\sigma(y) \in \gamma(\Gamma_i(y))$ for all $y \in \var{e}$
so that property (2) also holds.

\item[\textit{Rule \rulename{Let} is applied:}]
This is similar to rule \rulename{Free} since this rule
introduces the bound variable without any restriction on its value
so that we can ignore the evaluation of the bound expression.

\item[\textit{Rule \rulename{Case} is applied:}]
Then $e = \mathit{case}~x~\mathit{of}~\{p_1\to e_1; \ldots; p_n \to e_n\}$.
Since $\sigma(e)$ is reducible to $t$, there is some branch
$i \in \{1\ldots,n\}$ such that $\sigma(\rho(e_i))$
is reducible to $t$, where $\rho$ is a matching substitution on the
variables of pattern $p_i = c_i(\seq{x_{n_i}})$.
Let $\Gamma_i = \Gamma[x \mapsto c_i^{\alpha}(\seq{\top}), \seq{x_{n_i} \mapsto \top}]$.
Then $\var{e_i} \subseteq \dom{\Gamma_i}$,
$\Gamma_i \vdash e_i : \io_i$ is derivable where the height of
the proof tree is smaller than $h$,
and $\io_i =  \{ \seq{\Gamma'_{k'} \iomapsto a'_{k'}} \} \subseteq \io$.
By the induction hypothesis,
there is some $j \in \{1,\ldots,k'\}$ with
$\sigma(y) \in \gamma(\Gamma'_j(y))$ for all $y \in \var{e_i}$
and $t \in \gamma(a'_j)$.
This immediately implies property (1).
To show property (2),
consider the discriminating variable $x$ of the case expression,
which might not occur in $\var{e_i}$.
In this case, $\Gamma'_j(x) = \Gamma_i(x)$
since the abstract type of $x$ is not changed in the proof tree
for $\Gamma_i \vdash e_i : \io_i$ (compare proof of rule \rulename{Or}).
Since branch $i$ has been selected for the reduction of $\sigma(e)$,
$\sigma(x)$ must match the pattern $c_i(\seq{x_{n_i}})$ so that
$\sigma(x) \in \gamma(c_i^{\alpha}(\seq{\top})) = \gamma(\Gamma_i(x))$.
Finally, if there is some variable $y \in \var{e}$ with
$y \not\in \var{e_i}$,
then $\sigma(y) \in \gamma(\Gamma(y)) = \gamma(\Gamma'_j(y))$
since the abstract type of $y$ is not changed in the proof tree
for $\Gamma_i \vdash e_i : \io_i$ (compare proof of rule \rulename{Or}).
Therefore, $\sigma(y) \in \gamma(\Gamma'_j(y))$ holds for all
$y \in \var{e}$, which proves property (2).
\qed
\end{description}
\end{proof}
For an operation $f$ defined by $f(x_1,\ldots,x_n) = e$,
we infer its in/out type $\io$ by deriving the judgement
\[
\{ \seq{x_n \mapsto \top} \} \vdash e : \{ \seq{\Gamma_k \iomapsto a_k} \}
\]
and defining
\[
\io = \{ \Gamma_i(x_1) \cdots \Gamma_i(x_n) \iomapsto a_i \mid i=1,\ldots,k \}
\]
Hence, Theorem~\ref{theo:inout-rules-are-correct} implies
the following property of inferred in/out types.
\begin{corollary}[Correctness of inferred in/out types]
\label{cor:inout-rules-are-correct}
If
\[
\{ a_{i1} \cdots a_{in} \iomapsto a_i \mid i=1,\ldots,k \}
\]
is an inferred in/out type for an $n$-ary operation $f$
and $t_1,\ldots,t_n,t$ are data terms such that
$f(t_1,\ldots,t_n)$ is reducible to $t$,
then there is some $i \in \{1,\ldots,k\}$ with
$t_j \in \gamma(a_{ij})$ ($j = 1,\ldots,n$) and $t \in \gamma(a_i)$.
\end{corollary}

\subsection{Correctness of Call Types}
\label{appendix:calltype-correctness}

In the following we show that the inference rules
in Fig.~\ref{fig:calltypecheck} correctly approximate the
non-failure property of operations.

As described in Sect.~\ref{sec:func-call-type-checking},
if an operation $f$ is defined by $f(x_1,\ldots,x_n) = e$,
an assumed call type $\CT(f) = a_1 \ldots a_n$
is checked by deriving the judgement
\[
\{ (x_1, \{ \epsilon \mapsto a_1 \}, \epsilon), \ldots,
   (x_n, \{ \epsilon \mapsto a_n \}, \epsilon) \}, z = e \vdash \Delta
\]
for some fresh variable $z$.
If this check is successful, i.e., the judgement is derivable,
we say that $\CT(f)$ is \emph{verified}.
Furthermore, a (variable-free) expression $e$ is \emph{non-failing}
if all finite reductions of $e$ end in a data term (i.e., containing
only data constructors).
In the following we show that function calls are non-failing
for arguments satisfying $\CT$ if the call types of all functions
specified by $\CT$ are verified.

Recall that we write $\Delta(x)$ for the least upper bound
of all abstract type information about variable $x$ available
in the set of variable types $\Delta$, which is defined by
\[
\Delta(x) = \bigsqcup\; \{ a \mid
 (x, \io, x_1 \ldots x_n) \in \Delta,~ a_1 \ldots a_n \iomapsto a \in \io \}
\]
The inference rules derive sets of variable types $\Delta$
which should be satisfied by substitutions occurring in concrete
evaluation.
For this purpose, we define the \emph{domain} of a set of variable types
$\Delta$ by
\[
\dom{\Delta} = \{ x \mid (x, \io, x_1 \ldots x_n) \in \Delta \}
\]
We say that a set of variable types $\Delta$ is \emph{correct}
for a substitution $\sigma$
if, for all variables $x \in \dom{\Delta}$,
there is some $(x, \io, x_1 \ldots x_n) \in \Delta$
and some $a_1 \ldots a_n \iomapsto a \in \io$
such that $\sigma(x) \in \gamma(a)$ and
$\sigma(x_i) \in \gamma(a_i)$ for $i=1,\ldots,n$.

Thus, a set of variables types is correct for a substitution
if there is \emph{some} in/out type conform to the substitution.
This weakness is a consequence of our approximation
of non-deterministic computations.
Nevertheless, variables types can be helpful to deduce definite
information if they contain only a single element for some variable
or all elements have the same result type.

The following lemmas state simple consequences of this definition.
\begin{lemma}
\label{lemma:correct-variable-value}
If a set of variable types $\Delta$ is correct for a substitution $\sigma$,
then $\sigma(x) \in \gamma(\Delta(x))$ for all variables $x \in \dom{\Delta}$.
\end{lemma}
\begin{proof}
If $\Delta$ is correct for $\sigma$
and $x \in \dom{\Delta}$,
then there is some $(x, \io, x_1 \ldots x_n) \in \Delta$
and some $a_1 \ldots a_n \iomapsto a \in \io$ with $\sigma(x) \in \gamma(a)$.
Hence, by definition of $\Delta(x)$, $a \sqsubseteq \Delta(x)$
so that $\sigma(x) \in \gamma(a) \subseteq \gamma(\Delta(x))$.
\qed
\end{proof}
\begin{lemma}
\label{lemma:single-variable-correct}
If $\Delta = \{ (x, \{ \epsilon \mapsto a \}, \epsilon) \}$
and $\sigma(x) \in \gamma(a)$, then $\Delta$ is correct for $\sigma$.
\end{lemma}
\begin{proof}
This is an immediate consequence of the definition above.
\qed
\end{proof}
Now we can state the correctness of the inference rules
in Fig.~\ref{fig:calltypecheck} as follows.

\begin{theorem}
\label{theo:calltypecheck-rules-are-correct}
Assume that the call types of all functions specified by $\CT$ are verified.
Let $\var{e} \subseteq \dom{\Delta}$, $z \not\in \var{e}$,
$\Delta, z = e \vdash \Delta'$ be derivable by the inference rules
in Fig.~\ref{fig:calltypecheck},
and $\sigma$ be a substitution such that $\Delta$ is correct for $\sigma$.
Then all finite reductions of $\sigma(e)$ end
in some data term $t$ with $t \in \gamma(\Delta'(z))$
and $\Delta'$ is correct for $\sigma$.
\end{theorem}
\begin{proof}
We assume that the preconditions of the theorem hold, i.e.,
$\Delta, z = e \vdash \Delta'$ is derivable with a proof tree of height $h$
and $\sigma$ is a substitution such that $\Delta$ is correct for $\sigma$.
By Lemma~\ref{lemma:correct-variable-value},
$\sigma(x) \in \gamma(\Delta(x))$ for all $x \in \var{e}$.

We prove the claim by induction on the number of steps of
finite derivations of $\sigma(e)$
and a nested induction on the height of the proof tree.
We distinguish the kind of expression $e$:
\begin{description}
\item[Variable:] 
If $e = x$, where $x$ is a variable,
then $\sigma(x) \in \gamma(\Delta(x))$ is a data term
so that the evaluation of $\sigma(e)$ ends in $\sigma(x)$
(base case with $0$ derivation steps)
and $\sigma(x) \in \gamma(\Delta(x)) = \gamma(\Delta'(z))$
(by rule \vrulename{Var}).
With Lemma~\ref{lemma:single-variable-correct}, the claim holds.

\item[Constructor:] 
If $e = c(x_1,\ldots,x_n)$, where $x$ is a variable,
then $\sigma(x_i) \in \gamma(\Delta(x_i))$ are data terms
($i=1,\ldots,n$) so that $\sigma(e)$ is also a data term.
Hence, the evaluation of $\sigma(e)$ ends in the data term $t = \sigma(e)$
(base case with $0$ derivation steps).
Since
$\Delta' = \{\iosubst{z}
    {\{\top^n \iomapsto c^{\alpha}(\seq{\Delta(x_n)})\}}{x_1 \ldots x_n}\}$
(by rule \vrulename{Cons}),
$t = \sigma(e) \in \gamma(\Delta'(z))$ so that the claim holds
with Lemma~\ref{lemma:single-variable-correct}.

\item[Function:]
Let $e = f(x_1,\ldots,x_n)$, $f$ defined by $f(y_1,\ldots,y_n) = e'$,
and $\CT(f) = a_1 \ldots a_n$.
Since \vrulename{Func} is applicable,
$\sigma(x_i) \in \gamma(\Delta(x_i)) \subseteq \gamma(a_i)$
($i=1,\ldots,n$).
Let $\sigma' = \{ \seq{y_n \mapsto \sigma(x_n)} \}$.
Then $\sigma(e)$ is reducible to $\sigma'(e')$ and,
for all finite derivations of $\sigma(e)$, the same result
can be derived from $\sigma'(e')$ with a smaller number of derivations steps
so that we can apply the induction hypothesis to $\sigma'(e')$.
Let
\[
\Delta'' =
\{ (y_1, \{ \epsilon \mapsto a_1 \}, \epsilon), \ldots,
   (y_n, \{ \epsilon \mapsto a_n \}, \epsilon) \}
\]
Since $\sigma'(y_i) = \sigma(x_i) \in \gamma(a_i)$ for $i=1,\ldots,n$,
by Lemma~\ref{lemma:single-variable-correct},
$\Delta''$ is correct for $\sigma'$.
Since the call type of $f$ is verified,
$\Delta'', z' = e' \vdash \Delta'''$
is derivable for some fresh variable $z'$ and some set $\Delta'''$ of
variable types.
By the induction hypothesis, $\sigma'(e')$ is non-failing.
Let $t$ be some value of a derivation of $\sigma'(e')$.
The claim holds by Corollary~\ref{cor:inout-rules-are-correct}:
In particular, there is some $a_1 \cdots a_n \iomapsto a \in \IO(f)$
with $t \in \gamma(a)$.
Since $a \sqsubseteq \Delta'(z)$
(by rule \vrulename{Func} and definition of $\Delta'(z)$),
$t \in \gamma(\Delta'(z))$ so that the claim holds.

\item[Or expression:]
If $e = e_1~\mathit{or}~e_2$, then the first step of a derivation of
$\sigma(e)$ reduces to $\sigma(e_1)$ or $\sigma(e_2)$.
Assume the case $\sigma(e_1)$ (the other is symmetric).
Since $\Delta, z = e_1 \vdash \Delta_1$ is derivable
with a proof tree of height smaller than $h$,
the induction hypothesis implies that
$\sigma(e_1)$ is non-failing, $\Delta_1$ is correct for $\sigma$,
and if $t$ is some value of $\sigma(e_1)$, then $t \in \gamma(\Delta_1(z))$.
Since $\Delta \subseteq \Delta'$,
$\Delta_1(z) \sqsubseteq \Delta'(z)$ and the claim holds
(note that $\Delta_2$ is also correct for $\sigma$ by the induction
hypothesis applied to the proof tree of $\Delta, z = e_2 \vdash \Delta_2$).

\item[Free variables:]
If $e = \mathit{let}~x_1,\ldots,x_n ~\mathit{free~in}~ e'$,
then the first step of a derivation of $\sigma(e)$
reduces to $\sigma(\rho(e'))$ for some substitution
$\rho = \{ \seq{x_n \mapsto t_n} \}$ for the free variables.
By rule \vrulename{Free},
$\Delta \cup
 \{ (x_i, \{\epsilon \iomapsto \top\}, \epsilon) \mid i \in \{1,\ldots,n\} \},
 z = e \vdash \Delta'$
is derivable with a proof tree of height smaller than $h$.
By induction hypothesis (the extended set of variable types is trivially
correct for $\sigma \circ \rho$),
$\sigma(\rho(e'))$ is non-failing,
$\Delta'$ is correct for $\sigma \circ \rho$,
and if $t$ is some value of $\sigma(\rho(e'))$,
then $t \in \gamma(\Delta'(z))$.
This proves the claim.

\item[Let binding:]
Let $e = \mathit{let}~ x = e'~ \mathit{in}~e''$.
First, consider the evaluation of the bound expression $\sigma(e')$.
By rule \vrulename{Let},
$\Delta, x = e' \vdash \Delta'$ is derivable with a proof tree
of height smaller than $h$.
Thus, by the induction hypothesis, $\sigma(e')$ is non-failing
and $\Delta'$ is correct for $\sigma$.
Consider some value $t'$ of $\sigma(e')$
used to evaluate $\sigma(e)$ to $t$.
Then $t' \in \gamma(\Delta'(x))$ by the induction hypothesis.
Let $\sigma' = \{ x \mapsto t' \} \circ \sigma$.
Since $\Delta \cup \Delta', z = e'' \vdash \Delta''$
is derivable with a proof tree of height smaller than $h$,
by the induction hypothesis
(note that $\sigma'(y) \in \gamma((\Delta \cup \Delta')(y)$
for all $y \in \var{e''}$), $\sigma'(e'')$ is also non-failing,
$\Delta''$ is correct for $\sigma'$,
and, if $t''$ is some value of $\sigma'(e'')$, $t'' \in \gamma(\Delta''(z))$.
Thus, the claim holds.

\item[Case expression:]
If $e = \mathit{case}~x~\mathit{of}~\{p_1\to e_1; \ldots; p_n \to e_n\}$,
$\sigma(e)$ is evaluated by selecting some matching branch
(note that the discriminating argument $\sigma(x)$ cannot be
evaluated since $\sigma(x) \in \gamma(\Delta(x))$).
Such a selection step is possible since $\sigma(x) = c_i(y_1,\ldots,y_m)$
for some $i \in \{1,\ldots,n\}$
(since the case expression covers all constructors).
By rule \vrulename{Case},
$\Delta_i, z = e_i \vdash \Delta'_i$ is derivable with a proof tree
of height smaller than $h$, where
$p_i = c_i(x_1,\ldots,x_m)$ and
$\Delta_i = (\Delta \wedge [x \mapsto c_i])
            \cup \{ x_1 :: \top, \ldots, x_m :: \top \}$.
Since $(\Delta \wedge [x \mapsto c_i])$ constrains the abstract value
of $x$ to the abstraction of the actual value $\sigma(x)$,
$\Delta_i$ is correct for $\sigma \circ \rho$.
Thus, by the induction hypothesis,
$\sigma(\rho(e_i))$ is non-failing, where
$\rho$ is the matching substitution on the variables of pattern $p_i$,
and $\Delta'_i$ is correct for $\sigma \circ \rho$.
If $t$ is some value of $\sigma(\rho(e_i))$,
then $t \in \gamma(\Delta'_i(z))$ so that the claim holds
(since $\Delta'_i \subseteq \Delta'$).
\qed
\end{description}
\end{proof}
The previous theorem implies the following property of
successfully checked call types.
\begin{corollary}[Correctness of call type checking]
If, for all operations $f$ defined by $f(x_1,\ldots,x_n) = e$
and $\CT(f) = a_1 \ldots a_n$, the judgement
\[
\{ (x_1, \{ \epsilon \mapsto a_1 \}, \epsilon), \ldots,
   (x_n, \{ \epsilon \mapsto a_n \}, \epsilon) \}, z = e \vdash \Delta
\]
is derivable (for some fresh variable $z$),
then $f(t_1,\ldots,t_n)$ is non-failing if the
arguments satisfy the call types, i.e.,
$t_i \in \gamma(a_i)$ ($i=1,\ldots,n$).
\end{corollary}
\begin{proof}
Let
\[
\Delta_0 = \{ (x_1, \{ \epsilon \mapsto a_1 \}, \epsilon), \ldots,
              (x_n, \{ \epsilon \mapsto a_n \}, \epsilon) \}
\]
and $\sigma = \{ x_1 \mapsto t_1, \ldots, x_n \mapsto t_n \}$
where $t_i \in \gamma(a_i)$ ($i=1,\ldots,n$).
By Lemma~\ref{lemma:single-variable-correct},
$\Delta_0$ is correct for $\sigma$.
By Theorem~\ref{theo:calltypecheck-rules-are-correct},
all finite reductions of $\sigma(e)$ end in some data term
so that $f(t_1,\ldots,t_n)$ is non-failing.
\qed
\end{proof}

\end{document}